\documentclass[12pt]{article}
\newtheorem{theorem}{Theorem}
\newtheorem{corollary}[theorem]{Corollary}
\newtheorem{lemma}[theorem]{Lemma}

\newtheorem{remark}[theorem]{Remark}
\usepackage[toc,page]{appendix}
\newenvironment{proof}{\noindent{\bf Proof:}}{\hfill\fbox{}\vspace*{1mm}}
\usepackage[usenames]{color}
\usepackage{amssymb}
\usepackage{graphicx}
\newcommand{\pmu}{\mbox{\boldmath$\mu$}}

\RequirePackage[normalem]{ulem} 
\providecommand{\DIFdeltex}[1]{{\protect\color{red}\sout{#1}}}                      
\usepackage{xcolor}
\usepackage{changebar}
\newif\ifdiff
\difftrue
\ifdiff
  \newcommand{\del}[1]{\DIFdeltex{#1}}

\else
  \newcommand{\del}[1]{}

\fi

\begin{document}
\title{\bf On Correlated Defaults and Incomplete Information}
\author{
Wai-Ki Ching\thanks{Advanced Modeling and Applied Computing Laboratory,
Department of Mathematics, The University of Hong Kong,
Pokfulam Road, Hong Kong.
Email: wching@hku.hk.
}
\and Jia-Wen Gu$^*$\thanks{
$^*$Corresponding author, Department of Mathematical Sciences, University of Copenhagen, Universitetsparken 5,
DK-2100 Copenhagen, Denmark.
Email: jwgu.hku@gmail.com.}
\and Harry Zheng\thanks{Department of Mathematics,
Imperial College, London, SW7 2AZ, UK.
Email: h.zheng@imperial.ac.uk.
}
}

\maketitle

\begin{abstract}
In this paper, we study a continuous time structural asset value model
for two correlated firms using a two-dimensional Brownian motion.
We consider the situation of incomplete information,
where the information set available to the market participants
includes the default time of each firm and the periodic asset value reports.
In this situation, the default time of each firm becomes a
totally inaccessible stopping time to the market participants.
The original structural model is first transformed to a reduced-form model.
Then the conditional distribution of the default time together with
the asset value of each name are derived.
We prove the existence of the intensity processes of default times and
also give the explicit form of the intensity processes.
Numerical studies on the intensities
of the two correlated names are conducted for some special cases.
We also indicate the possible future research extension into three names case by considering a special correlation structure.
\end{abstract}


\thispagestyle{empty}

\section{Introduction}

There are two strands of literature for modeling credit risk,
namely, the structural firm value model originated by
Black and Scholes
\cite{BP} and Merton \cite{Merton}
and the reduced-form model
by Jarrow and Turnbull \cite{JT1,JT2}.
In the classical
firm value approach the asset value of a firm is described by a geometric Brownian motion
and the default is triggered when the asset value falls below a given default barrier level.
In the reduced-form intensity-based approach, defaults are modeled as exogenous events
and their arrivals are described by using random point processes and the default time
is a totally inaccessible stopping time.

{
The fundamental idea behind the paper dates back to the
seminal paper by Duffie and Lando [5] where a concrete example is given to show how a reduced-form model can be obtained from a structural model by restricting the information structure. They
consider a noisy and discretely observed firm asset value
in a continuous-time setting,
where the firm asset value is given by a diffusion process.
This idea and associated mathematical issue has been studied by Eillott, Jeanblanc and Yor  \cite{EJY}, Guo, Jarrow and Zeng  \cite{Guo}. It has been shown
in these papers that it is possible to transform a structural model with
a predictable default time into a reduced-form model,
with a totally inaccessible default time, by
introducing ``incomplete information''.

However, most of the papers focus on one dimensional diffusions,
while less attention has been placed to the correlation structure
in the incomplete information model with multiple firms.
Giesecke \cite{Giesecke} proposes a model of correlated multi-firm default with incomplete
information. Bond investors can observe issuers' assets and defaults, but not the default threshold
levels which are only revealed at times of default. Stochastic dependence between default events is
induced through correlated asset values and correlated default thresholds.

We study the correlation structure in an incomplete information framework
with asset values of firms being driven by
a two-dimensional correlated Brownian motion and with known default threshold levels.
The available information set
includes the periodic asset value reports and the default times which are
totally inaccessible stopping times to the market participants.
We can transform the original structural
model into a reduced-form intensity-based model and find the correlation
structure of the two firms under such a transformation.

The main contribution of this paper is that we investigate the transformation from a structural model to reduced-form model with ``incomplete information'' on multiple correlated assets,  provide the analytical expression for the conditional distribution of the default time, and prove the existence of the intensity process of default times, together with its explicit form.}
This paper sheds some new lights on the form of suitable intensity processes for portfolio credit risk.
Our findings show that it is
more reasonable to consider a contagion model, in which interacting intensities are time-varying with decaying default impacts (\cite{Gu3,Yu,ZJ}).

The remainder of the paper is organized as follows.
In Section 2 we present our model framework and give our main result: the default intensity
process characterizing the correlation structure of the two firms.
In Section 3 we illustrate our results by some numerical examples.
In Section 4 we give the proof of our main result and preliminary results concerning the conditional default
time and asset value distributions.
In Section 5 we give concluding remarks and point out further
research issues on the three-name case.

\section{The Model and the Main Result}\label{ModelResult}

Assume $(\Omega, P, \mathcal{F})$ is a given probability space.
The asset value process of firm $i$ is given by $V_i$, $i=1, 2$.
Define $\tau_i$ the default time of firm $i$ by
$$
\tau_i := \inf\{t>0: V_i(t) \leq B_i\}, \ i=1,2,
$$
where $B_i$ the default threshold value
of firm $i$ satisfying  $0< B_i < V_i(0)$. Denote by $X_i(t)=\ln \frac{V_i(t)}{B_i}$ for $t \geq 0$ and ${\bf X} = (X_1(t), X_2(t))^T$, ${\bf X}^T$ is the transpose of ${\bf X}$.
Assume ${\bf X}$ stisfies the SDE,
\begin{equation}\label{f1}
d{\bf X}(t)= \mbox{\boldmath$\mu$} dt+ \Sigma d{\bf W}(t),
\end{equation}
where ${\bf W}$ is a two-dimensional standard Brownian motion, $\pmu$ is a constant vector and $\Sigma$ is a constant matrix
given by
$$
\Sigma
=
\left(
\begin{array}{cc}
\sigma_1 \sqrt{1-\rho^2}& \sigma_1 \rho\\
0 & \sigma_2
\end{array}
\right),
$$
in which $\sigma_i >0$ is the volatility
of $X_i$, and $\rho$ is the correlation coefficient satisfying $|\rho| \leq 1$.
Note that $X_i(0) >0, i=1,2$.
By the continuity of $X_i$,
the default time of firm $i$ can be written equivalently as
$$
\tau_i=\inf\{t>0: X_i(t) = 0\}, \quad i=1,2.
$$
Assume investors receive periodic information,
at some fixed time instants
$$
0=t_0< t_1 < t_2 < \ldots < t_k< \ldots.
$$
The information available to investors up to time $t$ is given by
$$
\mathcal{F}_t:=\mathcal{F}_{X_1} \vee \mathcal{F}_{X_2} \vee \sigma(1_{\{\tau_1 \leq s\}}, 1_{\{\tau_2 \leq s\}}, 0 \leq s \leq t),
$$
where
$$
\mathcal{F}_{X_i}(t):=\sigma \left( X_i(t_j), j=0,1,\ldots, k_i
\quad {\rm with} \quad
t_{k_i} \leq \min \{t,\tau_i\}<t_{k_i+1}\right), \quad  i=1,2.
$$

{Even though $X_i$ is a diffusion process and $B_i$ is a constant, $\tau_i$ is a totally inaccessible stopping time as investors can only observe the information of $X_i$ at discrete times $t_k, k=0, 1, \ldots$, and not the information of $X_i$ at all time $t \geq 0$ (which would lead to $\tau_i$
being a predictable stopping time).
The main aim of this paper is to investigate the transformation from a structural model to a reduced-form model with this new information flow on two assets and deal with defaults with specified correlation structure.}

We reformulate the problem by using a transformation
proposed in Iyengar  \cite{Iyengar} and Metzler  \cite{Metzler} and the default times $\tau_i \ (i=1,2)$ can be redefined.
Let  { ${\bf Z}(t)=\Sigma^{-1}{\bf X}(t)$},
we have
$$
d{\bf Z}(t)={\bf m}dt+d{\bf W}(t),
$$
where
$
{\bf m} = \Sigma^{-1} \mbox{\boldmath$\mu$}.
$
The equivalent default times can be redefined as follows:
{
$$
\left\{
\begin{array}{ll}
 \tau_1&=\inf\left\{t>0: Z_1(t)=-\frac{\rho}{\sqrt{1-\rho^2}}Z_2(t)\right\},\\
\tau_2&=\inf\left\{t>0: Z_2(t)=0\right\}.
\end{array}
\right.
$$}
Denote by
$$
\alpha=\left\{
\begin{array}{ll}
\pi+\tan^{-1}\left(-\frac{\sqrt{1-\rho^2}}{\rho}\right),& \rho>0,\\
\frac{\pi}{2},& \rho=0,\\
\tan^{-1}\left(-\frac{\sqrt{1-\rho^2}}{\rho}\right),& \rho<0.
\end{array}
\right.
$$
We have $\alpha \in (0, \pi)$.
Define $${\tau:=\min(\tau_1, \tau_2)}.$$
Then $\tau$ is the first exit time of
${\bf Z}$ from the wedge
$$
\Omega_{\alpha}:=\{(r\cos \theta, r \sin \theta)^T: r>0, 0< \theta < \alpha\}
$$
with the initial position of ${\bf Z}$ given by
$$
{\bf z}_{t_0}:={\bf Z}(0)=\Sigma^{-1}{\bf X}(0)=(r_0 \cos \theta_0, r_0 \sin \theta_0)^T,
\quad{\rm where} \quad 0< \theta_0 <\alpha.
$$
The unconditional joint distributions of $({\bf Z}(\tau), \tau)$ and $(\tau_1, \tau_2)$ are
first derived in Iyengar \cite{Iyengar} and corrected in Metzler \cite{Metzler}.
The next theorem states their remarkable results.
\begin{theorem}\label{IyengarMetzler}

(Iyengar \cite{Iyengar} and Metzler \cite{Metzler})
{For ${\bf m}=(m_1, m_2)^T \in \mathcal{R}^2$,
\begin{equation}\label{p2}
P({\bf Z}(\tau) \in d{\bf z}, \tau \in dt)
= f(r, t, {\bf z}_{t_0}) dr dt,
\end{equation}
where ${\bf z} = (r \cos \alpha, r \sin \alpha)^T$ and
\begin{equation}\label{function f}
f(r,t, {\bf z}_{t_0})=\exp\left({\bf m}^{T}[(r \cos \alpha, r \sin \alpha)^T-{\bf z}_0]
-\frac{|{\bf m}|^2s}{2}\right)b(r,t, {\bf z}_{t_0})
\end{equation}
and
\begin{equation}\label{function b}
b(r, t, {\bf z}_{t_0})=
\frac{\pi}{\alpha^2 t r}\exp\left(-\frac{r^2+r_0^2}{2t}\right)
\sum_{n=1}^{\infty} n \sin \left(\frac{n \pi (\alpha-\theta_0)}{\alpha}\right)
I_{\frac{n \pi}{\alpha}}\left(\frac{r r_0}{t}\right)
\end{equation}
and $I_v$ is the modified Bessel function
of the first kind of order $v$, i.e.,
\begin{equation}\label{Bessel}
I_v(z)=\sum_{k=0}^{\infty}\frac{(\frac{1}{2}z)^{2k+v}}{k!  \Gamma(v+k+1)}.
\end{equation}}
For $s<t$, let
\begin{equation}\label{function g}
g(s, t, {\bf z}_{t_0})= \int_0^{\infty} f(r,s,{\bf z}_{t_0}) \pi(r\sin \alpha, t-s) dr,
\end{equation}
where
\begin{equation}\label{function pi}
\pi(x, h)= \frac{x}{\sqrt{2 \pi h^3}} \exp \left(-\frac{(x+m_2 h)^2}{2h}\right).
\end{equation}
We have
\begin{equation}
P(\tau_1 \in ds, \tau_2 \in dt)
=g(s, t, {\bf z}_{t_0}) ds dt.
\end{equation}
\end{theorem}

\begin{remark}\label{piFunction}
The function $\pi$ in Theorem \ref{IyengarMetzler} is the density of the hitting time to zero at time $h$ of a single Brownian motion with drift $m_2$ and
initial value $x$.
One can easily check by taking the partial derivative of $\pi(x, h)$ that,
for fixed $h>0$, $\pi(., h)$ achieves its maximum at $x=\frac{1}{2}\left(\sqrt{m_2^2h^2+4h}-m_2h\right)$
and for fixed $x>0$, $\pi(x, .)$ achieves its maximum at $h=\sqrt{\frac{x^2}{m_2^2}+\frac{9}{4m_2^4}}-\frac{3}{2m_2^2}$.
\end{remark}

    The conditional distributions of $\tau_i$ and $Z_i$,
can be derived and characterized with the given information flow, see Lemmas \ref{thm1} and \ref{thm2} in Section \ref{Proof}.

The main objective of this paper is to find the intensity process $\lambda_i$ of $\tau_i$, given the filtration $\mathcal{F}_u$.
The default time $\tau_i$ has an intensity process $\lambda_i$ with respect to the filtration $\mathcal{F}_u$
if $\lambda_i$ is a non-negative progressively measurable process satisfying
$$
\int_0^t \lambda_i(u) du < \infty
$$
a.s. for all $t$, such that
$$
1_{\{\tau_i \leq t\}}-\int_0^t \lambda_i(u) du, t\geq 0
$$
is an $\mathcal{F}_t$-martingale. Since
investors receive periodic asset information at times $t_k, k=0, 1, \ldots$, the values
${\bf z}_{t_k}:={\bf Z}(t_k)=(r_k \cos \theta_k, r_k \sin \theta_k)^T$ are known at $t_k$.

We can now state the main result of the paper.
\begin{theorem}\label{thm4}
For $t_j \leq u< t_{j+1}$, the default intensity process of $\tau_2$ exists and is given by
\begin{equation}\label{intensity}
\begin{array}{ll}
\lambda_2(u)=& \displaystyle 1_{\{\tau_2>u\}} \cdot 1_{\{\tau_1>u\}}
\left(\frac{\pi({ z}_{t_j}^{(2)},u-t_j)-\int_0^{u-t_j} g(s, u-t_j, {\bf z}_{t_j}) ds }{\int_{\Omega_{\alpha}} h(r, \theta, u-t_j, {\bf z}_{t_j}) dr d\theta}\right)\\
&\displaystyle +1_{\{t_j \leq s < u\}}\cdot 1_{\{\tau_2 >u\}} \cdot 1_{\{\tau_1 =s\}} \left(\frac{ g(s-t_j, u-t_j, {\bf z}_{t_j})}{\int_{u-t_j}^{\infty} g(s-t_j, t, {\bf z}_{t_j})dt}\right)\\
&\displaystyle +1_{\{\tau_2>u\}}\cdot 1_{\{\tau_1< t_j\}}\left(\frac{\pi({z}_{t_j}^{(2)},u-t_j)}{\int_{u-t_j}^{\infty} \pi({z}_{t_j}^{(2)},t) dt}\right),
\end{array}
\end{equation}
where ${ z}_{t_j}^{(2)}$ represents the second component of ${\bf z}_{t_j}$,
functions $g$ and $\pi$ are defined in (\ref{function g}) and (\ref{function pi}), respectively, and $h$
is defined by
\begin{equation}\label{function h}
\begin{array}{lll}
h(r, \theta, t, {\bf z}_j)
&=&\exp\left({\bf m}^T [(r \cos \theta, r \sin \theta)^T - {\bf z}_j]-\frac{|{\bf m}|^2t}{2}\right)\frac{2r}{t \alpha} \times\\
&& \exp \left( - \frac{r^2+r^2_j}{2t}\right) \sum_{n=1}^\infty \sin(\frac{n \pi \theta}{\alpha})\sin(\frac{n \pi \theta_j}{\alpha})I_{\frac{n \pi}{\alpha}}(\frac{r r_j}{t}).
\end{array}
\end{equation}
\end{theorem}

To derive the intensity process of default time $\tau_1$,
we perform the transformation $\tilde{\bf W}=\tilde{T}{\bf W}$,
$\tilde{\bf Z}=\tilde{T} \Sigma^{-1} {\bf X}$ and $\tilde{\bf m}=\tilde{T} \Sigma^{-1} \mbox{\boldmath$\mu$}$,
 where
$$
\tilde{T}=\left(
\begin{array}{cc}
-\rho & \sqrt{1-\rho^2}\\
\sqrt{1-\rho^2} & \rho
\end{array}
\right).
$$
Then we have
$$
d \tilde{\bf Z}=\tilde{\bf m} dt + d \tilde{\bf W}.
$$
The default times $\tau_i, i=1, 2$, are given by
{
$$
\left\{
\begin{array}{ll}
\tau_1 = \inf\{t> 0: \tilde{Z}_2(t) =0\}\\
\tau_2 = \inf\{t> 0: \tilde{Z}_1(t) =-\frac{\rho}{\sqrt{1-\rho^2}}\tilde{Z}_2(t)\}.
\end{array}
\right.
$$}
Following the same argument, we have
\begin{corollary}\label{cor1}
The intensity process of default time $\tau_1$ exists and is given by,
for $t_j \leq u< t_{j+1}$,
\begin{equation}\label{intensity1}
\begin{array}{lll}
\lambda_1(u) &=& \displaystyle 1_{\{\tau_1>u\}} \cdot 1_{\{\tau_2>u\}}
\left(\frac{\pi(\tilde{z}_{t_j}^{(2)},u-t_j)-\int_0^{u-t_j} \tilde{g}(s, u-t_j, \tilde{\bf z}_{t_j}) ds }{\int_{\Omega_{\alpha}} \tilde{h}(r, \theta, u-t_j, \tilde{\bf z}_{t_j}) dr d\theta}\right)\\
&&\displaystyle +1_{\{t_j \leq s < u\}}\cdot 1_{\{\tau_1 >u\}} \cdot 1_{\{\tau_2 =s\}} \left(\frac{ \tilde{g}(s-t_j, u-t_j, \tilde{\bf z}_{t_j})}{\int_{u-t_j}^{\infty} \tilde{g}(s-t_j, t, \tilde{\bf z}_{t_j})dt}\right)\\
&&\displaystyle +1_{\{\tau_1>u\}}\cdot 1_{\{\tau_2< t_j\}}\left(\frac{\pi(\tilde{z}_{t_j}^{(2)},u-t_j)}{\int_{u-t_j}^{\infty} \pi(\tilde{z}_{t_j}^{(2)},t) dt}\right),
\end{array}
\end{equation}
where $\tilde{g}$ and $\tilde{h}$ are the same
as $g$ and $h$, respectively, with ${\bf m}$ replaced by $\tilde{\bf m}$.
\end{corollary}

\begin{remark}
Theorem \ref{thm4} and Corollary \ref{cor1} show that the operational status of one firm can have substantial impact
to the default of another firm.
The default of one firm may cause  a significant change of default rate of another firm
for a certain period and its impact decreases after new information is released.
Starting from a reduced-form intensity model for portfolio credit risk, this theorem also shows that it is
more reasonable to consider a contagion model, in which interacting intensities are time-varying with decaying default impacts.
\end{remark}

The  proof of Theorem \ref{thm4} is given in Section \ref{Proof}. We first prove the case with single observation, and then extend to multiple observations.
We show the existence of the local default rate and, with the help of Aven's conditions \cite{Aven}, establish this local default rate being the intensity process.
The proof relies on the hitting time distribution of correlated Brownian motions,  see Iyengar \cite{Iyengar} and Metzler \cite{Metzler}.
We also derive the conditional distributions of default times and asset values as a by-product (Lemma \ref{thm1}, \ref{thm2}).
\section{Valuation and Approximation}\label{ValuationApproximation}

\subsection{Valuation in Some Special Cases and Implications}\label{valuation}
{It is not easy to find the intensities $\lambda_1$ and $\lambda_2$}
as the results in (\ref{intensity}) and (\ref{intensity1})
are complicated and
In this section, we give the valuation of $\lambda_1$ and $\lambda_2$
in some special cases and give some insights about
the correlation structure under incomplete information.
Note that Iyengar  \cite{Iyengar} and Metzler  \cite{Metzler} derive
the expressions (\ref{function g}) and (\ref{function h}) by solving a PDE.
When the correlation $\rho$ of the two Brownian motions satisfies certain conditions, the expressions in (\ref{function g}) and (\ref{function h}) can be simplified.
Let
$$
\rho_k=-\cos \left(\frac{\pi}{k}\right)
\quad {\rm and } \quad
\alpha_k=\frac{\pi}{k}, \quad k= 1, \ldots
$$
For a fixed $k$, $T_j$ is the matrix representing the reflection across the line
$$
y= x \tan(j \alpha_k) \quad {\rm and} \quad S_j=T_j S_{j-1}, \ \ S_0=I.
$$
Then we have

\begin{theorem}\label{thm5} (Iyengar \cite{Iyengar})
If ${\bf m} = {\bf 0}$, $\alpha=\alpha_k$ and ${\bf z} \in \Omega_{\alpha}$,
then
$$
P_{\{{\bf m} = {\bf 0}\}}({\bf Z}(t) \in d{\bf z}, \tau > t)
= \frac{1}{t} \sum_{j=0}^{2k-1}(-1)^j \Psi \left(\frac{z_0 -S_j {\bf z}}{\sqrt{t}} \right)d{\bf z},
$$
where $\Psi({\bf x})=(2\pi)^{-1} \exp(\frac{-{\bf x}^T {\bf x}}{2})$.
\end{theorem}

As claimed in Iyengar \cite{Iyengar}
and verified in Blanchet-Scalliet and Patras  \cite{BS},
$$
P_{\{{\bf m} = {\bf 0}\}}( {\bf Z}(t)\in d{\bf z}, \tau \in dt)
=\frac{1}{2}\frac{\partial}{\partial {\bf n}}P_{\{{\bf m} = {\bf 0}\}}({\bf Z}(t) \in d{\bf z}, \tau > t)
$$
where $\frac{\partial}{\partial {\bf n}}$ denotes the derivative in the direction of inward normal to the boundary $\partial \Omega_{\alpha}$
at point ${\bf z}$.
This simplifies the computation of $b(r, t, {\bf z}_{0})$ in (\ref{function b}) as follows:
\begin{equation}\label{bb}
\begin{array}{lll}
b(r, t, {\bf z}_{0})&=& \displaystyle \frac{1}{2t}\sum_{j=0}^{2k-1} (-1)^j
\frac{\partial}{\partial {\bf n}} \Psi \left(\frac{{\bf z}_0 -S_j {\bf z}}{\sqrt{t}} \right)\\
&=&\displaystyle  \frac{1}{2t}\sum_{j=0}^{2k-1} (-1)^j \frac{\partial}{\partial {\bf n}^{\prime}} \Psi \left(\frac{\tilde{\bf z}_0 -S_j {\bf z}}{\sqrt{t}} \right)\\
&=& \displaystyle \frac{1}{2t^2}\sum_{j=0}^{2k-1} (-1)^j \Psi \left(\frac{\tilde{\bf z}_0 -S_j \tilde{\bf z}}{\sqrt{t}} \right)({\tilde{\bf z}_0}^T S_j {\bf e}_2),
\end{array}
\end{equation}
where ${\bf n}=(\sin\alpha, -\cos \alpha)$, ${\bf n}^{\prime}=(0, 1)$, $\tilde{\bf  z}_0=(r_0 \cos\tilde{\theta}_0, r_0 \sin\tilde{\theta}_0)$ with $\tilde{\theta}_0=\alpha-\theta_0$ and $\tilde{\bf z}=(r,0)$.
We can also get a simple expression for h, defined in (\ref{function h}), as
\begin{equation}\label{hh}
h(r, \theta, t, {\bf z}_0)
=\exp\left({\bf m}^T [(r \cos \theta, r \sin \theta)^T - {\bf z}_0]-\frac{|{\bf m}|^2t}{2}\right)
\frac{r}{t} \sum_{j=0}^{2k-1}(-1)^j \Psi \left(\frac{{\bf z}_0 -S_j {\bf z}}{\sqrt{t}} \right).
\end{equation}
By combining Eq.s (\ref{intensity}), (\ref{intensity1}), (\ref{bb}), and (\ref{hh}), we conduct a numerical study to investigate,
in the special cases ($\alpha=\alpha_k$),
the interacting intensities of the two firms.

\begin{remark}
We have $\rho_k \rightarrow -1$ as $k \rightarrow \infty$, which implies that when $\rho$ is sufficiently close to $-1$, one can approximate
the intensity process by that of $\rho_k$, where $\rho_k \leq \rho < \rho_{k+1}$.
\end{remark}

We next give two numerical examples to illustrate the infection effect of the two firms.
The data used are
$
{\bf x}_0=(9, 10)^T, \ \mbox{\boldmath$\mu$}=(2,3)^T, \ \sigma_1=4, \ {\rm and}  \ \sigma_2=5.
$

Figure 1 presents intensity process $\lambda_2$ of default time $\tau_2$ of firm B in independent case ($\rho=0$), negative correlated case ($\rho=-0.5$) and highly negative correlated case ($\rho=-0.7$),
where the time length between two observations ($\Delta t$) is $10$ (years) and the default time of firm A is $\tau_1=2$.
Note that when $\rho=0$, $\lambda_2$ is not affected by the default event of firm A, which is consistent with our expectation as firm A and B are independent of each other.
The default of firm A has a substantial impact on
the intensity process $\lambda_2$ when $\rho \neq 0$, which drops
sharply as the two firms are negative correlated.
The default of firm A has a even more significant
impact on firm B when they are highly negative correlated.
The figure also reveals a fact that before a default is observed,
the intensity process $\lambda_2$ in the correlated case is nearly the
same as that in the independent case.

\begin{figure}
	\centering
		{\includegraphics[scale=0.6]{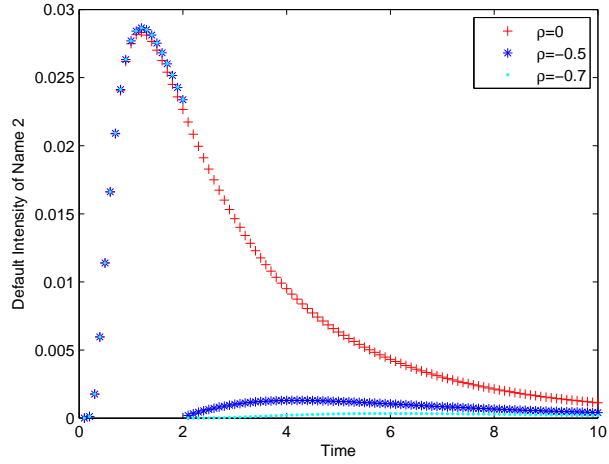}}
	\label{fig:intent234a}
	\caption{Default Intensity process $\lambda_2$ when $\tau_1=2$}\label{figure1}
\end{figure}

Figure 2 presents $\lambda_2$ with varying default times of firm A with the same data and $\rho=-0.5$. We observe that as two firms are negative correlated, $\lambda_2$ experiences a sharp drop at the default time of firm A, after which $\lambda_2$ increases slightly as the default impact decreases.
In the long run, $\lambda_2$ tends to zero as the drift is positive.

\begin{figure}\label{figure3}
	\centering
		{\includegraphics[scale=0.6]{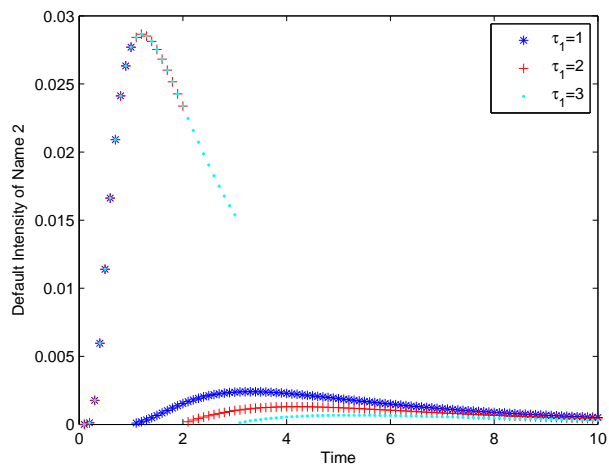}}
	\label{fig:intent234}
	\caption{Default Intensities of firm B where $\rho=-0.5$}
\end{figure}

\subsection{Approximation in General Cases}\label{Approxiamtion}
We now focus on the interval $[t_0, t_1]$.
To simplify the notations, we write $g(s, t, {\bf z}_{0})$ as $g(s, t)$
and $h(r, \theta, t, {\bf z}_{0})$ as $g(r, \theta, t)$ in the rest of the paper.
The evaluation of the intensity of each firm by taking the integration of $g(s, t)$ and $h(r, \theta, t)$
{ incurs a large computational cost.}
We employ
the numerical method in Kou and Zhong  \cite{Kou}, Rogers and Shepp  \cite{Rogers}
to evaluate the Laplace transform
of $(\tau_1, \tau_2)$ and $\tau$.
Using the inverse Laplace transform, we can find the distributions of $(\tau_1, \tau_2)$ and $\tau$.
The terms in the expression of $\lambda_2(t)$ in Theorem \ref{thm3}
can be expressed as:
$$
\left\{
\begin{array}{l}
\displaystyle
\int_0^u g(s, u) ds = \frac{d P(\tau_1 \leq s, \tau_2 \leq t)}{dt}\left|_{(s,t)=(u, u)}\right.\\
\displaystyle \int_u^{\infty} g(s, v) dv = \frac{d P(\tau_1 \leq s, \tau_2 > u)}{ds},\\
\displaystyle g(s, u) = \frac{d^2 P(\tau_1 \leq s, \tau_2 \leq t)}{ ds dt},\\
\displaystyle  \int_{\Omega_{\alpha}} h(r, \theta, u) dr d\theta = P(\tau > u).\\
\end{array}
\right.
$$
The Laplace transform of the joint distribution of $(\tau_1, \tau_2)$
is given by
\begin{equation}\label{laplace}
E[e^{-p_1\tau_1-p_2\tau_2} \mid {\bf X}(0)=(x_1, x_2)].
\end{equation}
The solution to the following PDE
$$
\mathcal{L}(u)(x_1, x_2)=(p_1+p_2)u
$$
where
$$
\mathcal{L}=\frac{1}{2}\sigma_1^2 \frac{\partial^2}{\partial x_1^2}+\rho \sigma_1 \sigma_2\frac{\partial^2}{\partial x_1 \partial x_2}+\frac{1}{2}\sigma_2^2 \frac{\partial^2}{\partial x_2^2}+\mu_1\frac{\partial}{\partial x_1}+\mu_2\frac{\partial}{\partial x_2}
$$
if exists, has the representation given in Eq. (\ref{laplace}).
In their innovative papers \cite{Kou, Rogers}, Kontorovich-Lebedev transform
and finite Fourier transform
are proposed to solve the PDE for (\ref{laplace}).
We use these methods to solve the PDE and find the joint distribution numerically.
We can then approximate the intensity process for general $\rho$.

Figure 3 presents the approximate and the exact intensity process $\lambda_2$ in independent case.
Data used are
$
{\bf x}_0=(1, 1)^T, \ \mbox{\boldmath$\mu$}=(0.1, -0.2)^T, \ \sigma_1=1.2,
\ {\rm and} \ \sigma_2=0.5.
$
Numerical tests show that the approximation method
can give a good approximation of the default intensity process.

\begin{figure}\label{figure4}
	\centering
		{\includegraphics[scale=0.6]{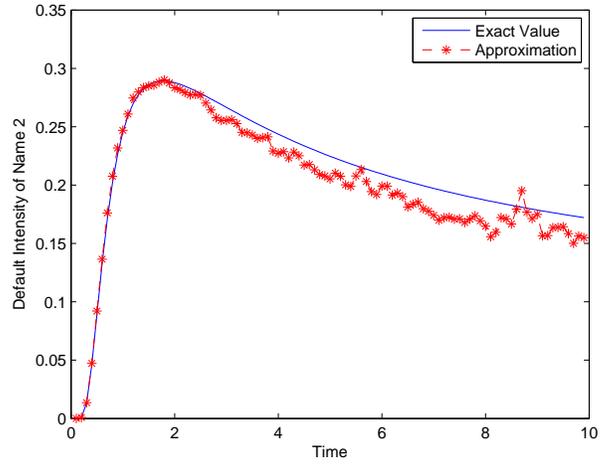}}\label{c1}
	\caption{Default Intensity Process $\lambda_2$ when $\rho=0$}
\end{figure}

To give examples of the default intensity in general cases, we adopt the same parameters and assume the correlated parameter $\rho$
equals $0.1$  and $-0.1$.
From Figures 5 and 6, we can observe a sharp upward jump in $\lambda_2$ when firm A defaults at time $\tau_1=2$ if $\rho=0.1$ and a sharp downward jump if $\rho=-0.1$.
This is consistent with our intuition.

\begin{figure}\label{figure5}
	\centering
		{\includegraphics[scale=0.6]{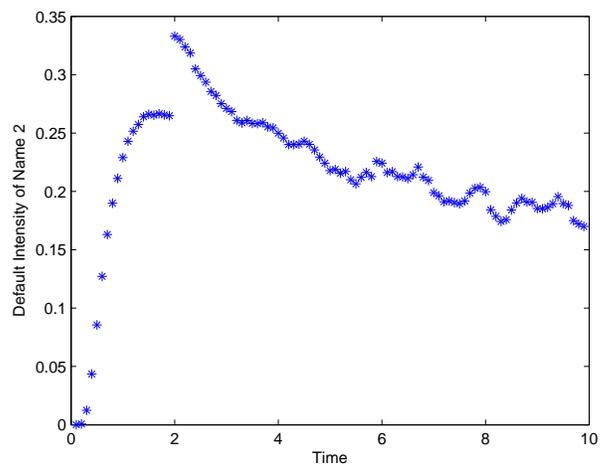}}
	\caption{Default Intensity Process $\lambda_2$ when $\tau_1=2$ and $\rho=0.1$}
\end{figure}

\begin{figure}\label{figure6}
	\centering
		{\includegraphics[scale=0.6]{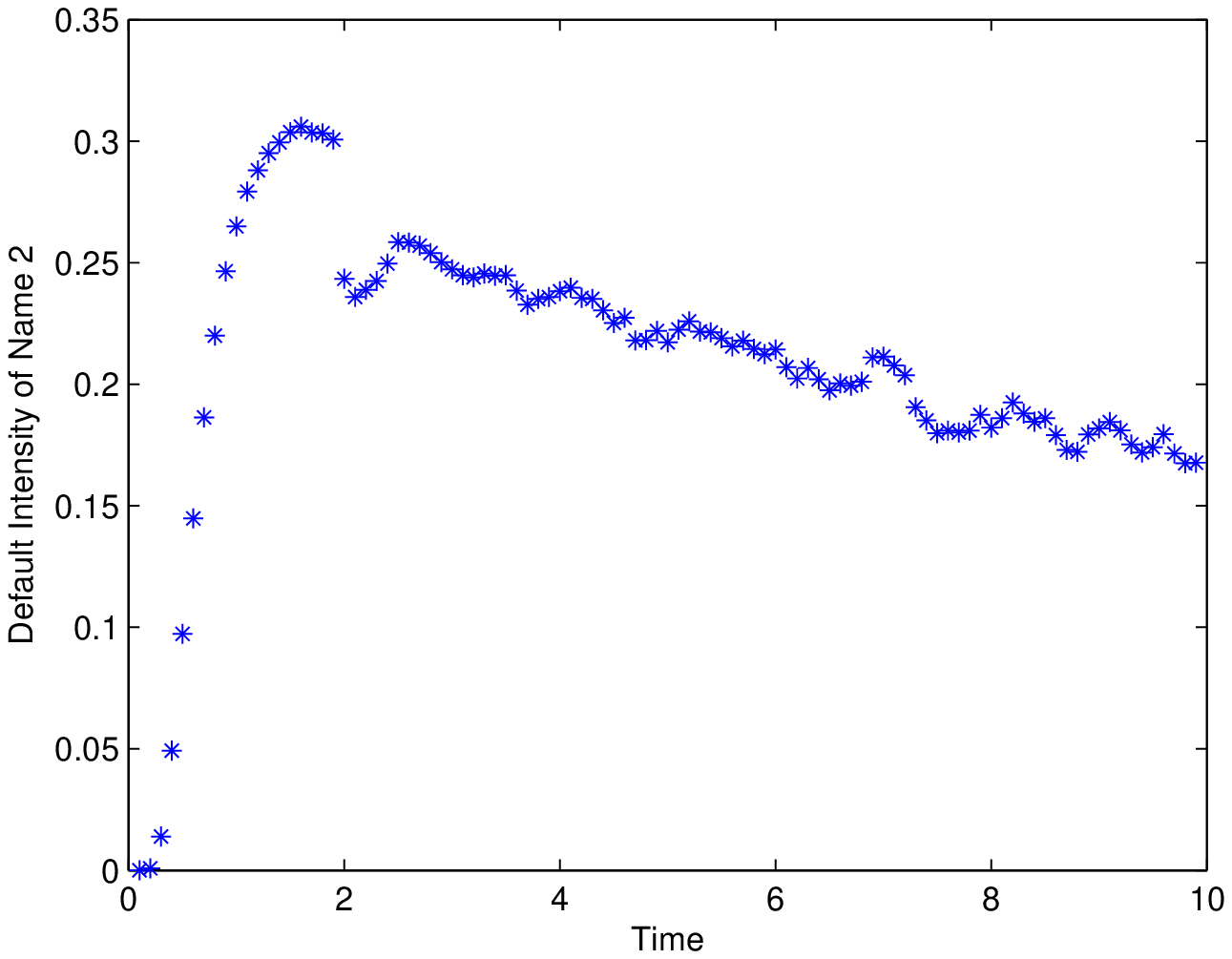}}
	\caption{Default Intensity Process $\lambda_2$ when $\tau_1=2$ and $\rho=-0.1$}
\end{figure}

{
\section{Proof of Theorem \ref{thm4}}\label{Proof}
\subsection{Preliminary Results}\label{Preliminary}

The objective of this subsection is to
obtain the conditional distributions of $\tau_i$ and $Z_i$ ($i=1, 2$)
given $\mathcal{F}_u$.
For simplicity, we assume that $t_0<u<t_1$.
We only need to derive the conditional distribution of $\tau_2$ given $\mathcal{F}_u$,
that of $\tau_1$ can be derived easily with a transformation.
The derivation relies on the results given in  Theorem \ref{IyengarMetzler}.

We now  give the conditional default time distribution given $\mathcal{F}_u$.
\begin{lemma}\label{thm1}
For $0<s< u<v$, the conditional distribution of $\tau_2$ is given by,
\begin{equation}
\begin{array}{ll}
1_{\{\tau_2 > u\}} P(\tau_2 \in dv \mid \mathcal{F}_u)=& 1_{\{\tau_2 > u\}} \cdot 1_{\{\tau_1 > u\}}
\displaystyle \left(\frac{\pi(z_0^{(2)}, v)-\int_0^u g(s,v) ds}{\int_{\Omega_{\alpha}} h(r, \theta, u) dr d\theta} dv\right)\\
&+1_{\{\tau_2 > u\}} \cdot 1_{\{\tau_1 =s\}} \displaystyle \left(\frac{g(s, v)}{\int_u^{\infty} g(s, t) dt} dv\right),
\end{array}
\end{equation}
where $z_0^{(2)}$ is the second entry of ${\bf z}_0$ and
$h(r,\theta, u)$ is given in Eq.(\ref{h}).
\end{lemma}

\begin{proof}
By Bayes' rule, for $s< u<v$, the conditional distribution of $\tau_2$ is given by
\begin{equation}\label{cp}
\begin{array}{ll}
1_{\{\tau_2 > u\}} \cdot P(\tau_2 \in dv \mid \mathcal{F}_u)=&
1_{\{\tau_2 > u\}} \cdot 1_{\{\tau_1 > u\}} \cdot P(\tau_2 \in dv \mid \tau > u) \\
&+1_{\{\tau_2 > u\}} \cdot 1_{\{\tau_1 =s\}} \cdot P(\tau_2 \in dv \mid \tau_2 > u, \tau_1=s).
\end{array}
\end{equation}
For the first term of the RHS of  Eq. (\ref{cp}), we have
$$
\begin{array}{lll}
P(\tau_2 \in dv \mid \tau >u) &=& \displaystyle \frac{P(\tau_2 \in dv, \tau_1 > u)}{P(\tau >u)} \\
&=& \displaystyle \frac{P(\tau_2 \in dv)-P(\tau_2 \in dv, \tau_1 \leq u)}{P(\tau >u)} \\
&=& \displaystyle \frac{\pi(z_0^{(2)}, v)-\int_0^u g(s,v) ds}{P(\tau >u)} dv.
\end{array}
$$
From Iyengar \cite{Iyengar} and Metzler \cite{Metzler},
we also have
$$
P(\tau >u)= \int_{\Omega_{\alpha}} h(r, \theta, u) dr d\theta,
$$
where
\begin{equation}\label{h}
\begin{array}{lll}
h(r, \theta, t)
&=&\displaystyle \frac{2r}{t \alpha}\exp \left(m_1(r\cos \theta-z_0^{(1)})+m_2(r\sin \theta-
z_0^{(2)})-\frac{|{\bf m}|^2t}{2}\right)  \exp \left(-\frac{r^2+r_0^2}{2t}\right) \times \\
& &\displaystyle \sum_{n=1}^{\infty} \sin \left(\frac{n \pi \theta}{\alpha} \right)\sin \left( \frac{n \pi \theta_0}{\alpha}\right)
I_{\frac{n \pi}{\alpha}}\left(\frac{r r_0}{t}\right).
\end{array}
\end{equation}
Note that the second term of the RHS of Eq. (\ref{cp}) is
$$
P(\tau_2 \in dv \mid  \tau_2 > u, \tau_1=s) = \displaystyle \frac{P(\tau_2 \in dv, \tau_1 \in ds)}{P(\tau_2 >u, \tau_1\in ds)}
= \displaystyle \frac{g(s, v)}{\int_u^{\infty} g(s, t) dt} dv.
$$
The result follows.
\end{proof}

We also present the conditional distribution of $Z_i(u), i=1,2$ given the filtration $\mathcal{F}_u$.
Again we give the conditional distribution of $Z_2(u)$,
while the conditional distribution of $Z_1(u)$ can be derived similarly.
\begin{lemma}\label{thm2}
For $0< s<u$,
\begin{equation}
\begin{array}{lll}
1_{\{\tau_2>u\}}P(Z_2(u) \in dx \mid \mathcal{F}_u)&=&1_{\{\tau_2>u\}}\cdot 1_{\{\tau_1 >u\}}
\displaystyle \left(\frac{\tilde{\pi}(x,z^{(2)}_0, u)-p(x,u)}{\int_{\Omega_{\alpha}} h(r, \theta, u) dr d\theta} dx\right)\\
&&+1_{\{\tau_2>u\}} \cdot 1_{\{\tau_1 =s\}} \displaystyle \left(\frac{l(s,u,x)}{\int_u^{\infty} g(s, t) dt} dx\right),
\end{array}
\end{equation}
where
$$
\left\{
\begin{array}{lll}
l(s,t,x)&=&\displaystyle \int_0^{\infty} f(r, s) \tilde{\pi}(x, r \sin \alpha, u-s) dr,\\
p(x, t)&=&\displaystyle \int_0^t l(s,t,x) ds,\\
\tilde{\pi}(x, x_0, h)&=& \displaystyle \frac{1}{\sqrt{2 \pi h}} \exp\left(-\frac{(x-x_0-m_2 h)^2}{2h}\right)\left(1-\exp(-\frac{2 x_0 x}{h})\right).
\end{array}
\right.
$$
\end{lemma}

\begin{proof}
{
By Bayes' rule, one can obtain
\begin{equation}\label{Z}
\begin{array}{ll}
1_{\{\tau_2>u\}}P(Z_2(u) \in dx \mid \mathcal{F}_u)=&1_{\{\tau_2>u\}} \cdot 1_{\{\tau_1 >u\}}P(Z_2(u) \in dx \mid \tau >u)\\
&+1_{\{\tau_2>u\}}\cdot 1_{\{\tau_1 =s\}} P(Z_2(u) \in dx \mid \tau_2 >u, \tau_1=s).
\end{array}
\end{equation}
We note that according to the density function of the first passage time
and {a straightforward calculation
(see for instance Chapter 1 of Harrison \cite{Harrison}),}
the probability of a Brownian motion with drift $m_2$, conditioning on starting at some level $x_0$ at time $0$,
ends at level $x$ at time $h$ with running minimum being positive is
$$
\tilde{\pi}(x, x_0, h)= \displaystyle \frac{1}{\sqrt{2 \pi h}} \exp\left(-\frac{(x-x_0-m_2 h)^2}{2h}\right)\left(1-\exp(-\frac{2 x_0 x}{h})\right).
$$
For the second term of Eq. (\ref{Z}),
$$
 P(Z_2(u) \in dx \mid \tau_2 >u, \tau_1=s) = \displaystyle \frac{P(\tau_2>u, Z_2(u) \in dx, \tau_1 \in ds)}{P(\tau_1 \in ds, \tau_2 >u)},
$$
and
$$
\begin{array}{lll}
P(\tau_2>u, Z_2(u) \in dx, \tau_1 \in ds)&=&\displaystyle \int_0^{\infty} f(r,s) P(\tau_2>u, Z_2(u) \in dx \mid Z_2(s)=r \sin \alpha) dr ds\\
&=&\displaystyle \int_0^{\infty} f(r,s) \tilde{\pi}(x, r\sin\alpha, u-s) dr dx ds\\
&=&l(s,u,x)ds dx.
\end{array}
$$
For the first term of Eq. (\ref{Z}),
$$
\begin{array}{lll}
P(Z_2(u) \in dx \mid \tau >u) &=&\displaystyle \frac{P(\tau>u, Z_2(u) \in dx)}{P(\tau >u)}\\
&=&\displaystyle \frac{P(\tau_2>u, Z_2(u) \in dx)-P(\tau_1 \leq u, \tau_2 >u, Z_2(u)\in dx)}{P(\tau >u)},
\end{array}
$$
where
$$
P(\tau_2>u, Z_2(u) \in dx) =\tilde{\pi}(x, z^{(2)}_0,u) dx,
$$
and
$$
P(\tau_1 \leq u, \tau_2 >u, Z_2(u)\in dx)=\int_0^{u}l(s, u,x) ds dx=p(x,u) dx.
$$
The the result follows.}
\end{proof}

Given the survival to $u$,
the above theorem gives us the conditional distribution of assets,
because the conditional density of ${\bf V}(t)$ can be easily obtained from the
conditional density of ${\bf Z}(t)$.

\subsection{Existence and Explicit Form of $\lambda_i$}\label{ExistenceLambda}
To prove Theorem {\ref{thm4}}, we first begin with single observations, i.e. $t_0 < u <t_1$, and then extend to multiple observations.
The intuitive meaning of the intensity is given by the local default rate,
\begin{equation}\label{lim}
\lim_{\delta u \rightarrow 0} \frac{1}{\delta u} P(\tau_i \in (u, u+\delta u] \mid \mathcal{F}_u).
\end{equation}}
To obtain the intensity process $\{\lambda_i(u)\}_{u \geq 0}$,
we need the following lemmas.
\begin{lemma}\label{gcontinuous}
For a fixed $s$, $g(s, .)$ is continuous in $(s, +\infty)$.
\end{lemma}

\begin{proof}
By the following elementary inequalities
$$
\sqrt{a+b} \pm \sqrt{a}\leq \sqrt{2(2a+b)}, \ \ a>0, \ \  b>0,
$$
and Remark \ref{piFunction}, for a fixed $h>0$, $\pi(.,h)$ achieves its local maximum at
$$
x=\frac{1}{2}\left(\sqrt{m^2_2 h^2+4h}-m_2h\right).
$$
{We have}
\begin{equation}\label{pi}
\pi(x, h)=\frac{x}{\sqrt{2 \pi h^3}} \exp\left(-\frac{(x+m_2h)^2}{2h}\right) \le
\frac{\frac{1}{2}\left(\sqrt{m^2_2 h^2+4h}-m_2h\right)}{\sqrt{2 \pi h^3}}
\leq \frac{\sqrt{m^2_2h+2}}{\sqrt{2 \pi  h^2}}.
\end{equation}
For any $t >s$ and $n \in \mathcal{N}^+$,
$$
g\left(s, t+\frac{1}{n}\right)
=\int_0^{\infty} f(r, s) \pi\left(r \sin \alpha, t+\frac{1}{n}-s\right) dr,
$$
where
$$
\begin{array}{lll}
f(r, s) \pi(r \sin \alpha, t+\frac{1}{n}-s)
&\leq&
\displaystyle f(r,s)\cdot \frac{\sqrt{m^2_2 (t+\frac{1}{n}-s)+2}}{\sqrt{2 \pi  (t+\frac{1}{n}-s)^2}}\\
&\leq& \displaystyle f(r,s) \cdot \frac{\sqrt{m^2_2 (t+1-s)+2}}{\sqrt{2 \pi (t-s)^2}}.
\end{array}
$$
Since $f$ gives the joint density of $(\tau, {\bf Z}(\tau))$,
we have
$$
\int_0^{\infty} f(r, s) dr < \infty.
$$
Using the Dominated Convergence Theorem (DCT) and
the fact that $\pi$ is continuous in $(0, \infty)\times(0, \infty)$, we have
$$
\begin{array}{lll}
\displaystyle \lim_{n \rightarrow \infty}
g\left(s, t+\frac{1}{n}\right) &=& \displaystyle \int_0^{\infty}  \lim_{n \rightarrow \infty} f(r,s)\pi\left(r \sin \alpha, t+\frac{1}{n}-s\right) dr\\
&=& \int_0^{\infty}f(r,s)\pi(r \sin \alpha, t-s) dr\\
&=& g(s, t).
\end{array}
$$
\end{proof}

\begin{lemma}\label{Bessellemma}
{
For modified Bessel function of the first kind of order $v$, $I_v$
(c.f. Eq. ({\ref{Bessel}) ), we have for $z \in (0,1)$,
$$
\sum_{n=1}^{\infty} n  I_{\frac{n \pi}{2 \alpha}}(z) \leq C z^{\frac{\pi}{2 \alpha}},
$$
where we denote by $C$ a generic positive constant, not necessarily the same at each appearance in the rest of the paper.

}}
\end{lemma}
\begin{proof}
For any positive real number $v$,
$$
I_v(z)=\sum_{k=0}^{\infty}\frac{(\frac{1}{2}z)^{2k+v}}{k!  \Gamma(v+k+1)} \leq \sum_{k=0}^{\infty}\frac{(\frac{1}{2}z)^{2k+v}}{k!} = e^{\frac{z^2}{4}} \left(\frac{1}{2}z\right)^{v}  <e \left(\frac{1}{2}z\right)^{v}.
$$
Therefore,
$$
\sum_{n=1}^{\infty} n I_{\frac{n \pi}{2 \alpha}}(z) < e \sum_{n=1}^{\infty} n\left(\frac{1}{2} z\right)^{\frac{n\pi}{2 \alpha}}\leq \frac{e(1/2)^{\frac{\pi}{2 \alpha}} z^{^{\frac{\pi}{2 \alpha}}}}{[1-(z/2){^{\frac{\pi}{2 \alpha}}}]^2} \leq \frac{e(1/2)^{\frac{\pi}{2 \alpha}} z^{^{\frac{\pi}{2 \alpha}}}}{[1-(1/2){^{\frac{\pi}{2 \alpha}}}]^2}=C z^{\frac{\pi}{2 \alpha}}.
$$
\end{proof}

The next result characterizes the local default rate.
\begin{lemma}\label{localRate}
For $0<s<u$ and $\{u_n\}$ a sequence decreasing to $0$, we have
$$
\begin{array}{ll}
& \displaystyle \lim_{n \rightarrow \infty}
\frac{1_{\{\tau_2 >u\}}}{u_n}P(\tau_2 \in (u, u+u_n] \mid \mathcal{F}_u)\\
=&
\displaystyle 1_{\{\tau_2>u\}}\cdot 1_{\{\tau_1>u\}} \left(\frac{\pi(z_0^{(2)},u)-\int_0^u g(s, u) ds }{\int_{\Omega_{\alpha}} h(r, \theta, u) dr d\theta}\right)
+1_{\{\tau_2 >u\}}\cdot 1_{\{\tau_1 =s\}} \left(\frac{ g(s, u)}{\int_u^{\infty} g(s, v)dv}\right).
\end{array}
$$
\end{lemma}
\begin{proof}
{
By Theorem \ref{thm1}, for $s<u$, we have
\begin{equation}\label{nn}
\begin{array}{lll}
1_{\{\tau_2 > u\}}\cdot P(\tau_2 \in (u, u+u_n] \mid \mathcal{F}_u)=& \displaystyle
1_{\{\tau_2>u\}} \cdot 1_{\{\tau_1>u\}} \cdot
\left( \frac{\int_u^{u+u_n}\pi(z_0^{(2)},v)dv-\int_u^{u+u_n}\int_0^u g(s, v) ds dv}{\int_{\Omega_{\alpha}} h(r, \theta, u) dr d\theta}\right)\\
&\displaystyle + 1_{\{\tau_2 >u\}} \cdot 1_{\{\tau_1 =s\}}\cdot \left(\frac{\int_u^{u+u_n} g(s, v) dv}{\int_u^{\infty} g(s, t)dt}\right).
\end{array}
\end{equation}
For the first term of the RHS of Eq. (\ref{nn}), we have
\begin{equation}\label{contin}
\lim_{n \rightarrow \infty}\frac{1}{u_n} \int_u^{u+u_n}\int_0^u g(s, v) ds dv = \int_0^u g(s, u) ds.
\end{equation}

{Eq. (\ref{contin}) can be proved as follows.}
Since $g$ is non-negative a.s.,
Fubini's Theorem tells us that
\begin{equation}
\begin{array}{lll}
\displaystyle \frac{1}{u_n}\int_u^{u+u_n}\int_0^u g(s, v) ds dv
&=&\displaystyle \int_0^{u}\int_0^{\infty} f(r,s)\frac{1}{u_n}
\left(\int_u^{u+u_n}\pi(r \sin \alpha, v-s) dv\right) dr ds\\
&=&\displaystyle \int_0^{u}\int_0^{\infty} f(r,s)\pi(r \sin \alpha, u+\eta_n-s)dr ds\\
&=&\displaystyle \int_0^u g(s, u+\eta_n) ds,
\end{array}
\end{equation}
where $0< \eta_n< u_n$.
Let $u^* \in (0, u)$ and $0 < s<u<t$.
For $s \in (0, u^*)$, { by Eq. (\ref{pi})}
{
\begin{equation}\label{firin}
g(s, t)  \leq  \frac{\sqrt{m^2_2(t-s)+2}}{\sqrt{2 \pi} (t- s)}\int_0^{\infty} f(r,s) dr
 \leq \frac{\sqrt{m^2_2t+2}}{\sqrt{2 \pi} (u- u^*)}\int_0^{\infty} f(r,s) dr.
\end{equation}
We can select $u^*, \eta^*$ to be sufficiently close to $u$, such that for $t \in (u, \eta^*)$ and $s \in [u^*, u)$
$$
\frac{r_0^2}{2s}\left(\frac{t-s}{t-s\cos 2\alpha }\right) \in (0, 1).
$$
This can be done since
$$
t-s\cos 2\alpha = t-s+2s\sin^2 \alpha \geq 2u^* \sin^2 \alpha >0.
$$
Then we claim that
\begin{equation}\label{ineq}
\begin{array}{lll}
g(s,t)&\leq & \displaystyle \frac{C}{\sqrt{s}\sqrt{t-s\cos 2\alpha}(t-s)}\exp\left(\frac{|{\bf m}|^2s}{2}\right)
\sum_{n=1}^{\infty} n \sin \left(\frac{n \pi \tilde{\theta}_0}{\alpha}\right) I_{\frac{n \pi}{2\alpha}}\left(\frac{r_0^2}{2s}\frac{t-s}{t-s\cos 2\alpha }\right)\\
& \leq & \displaystyle \frac{C}{\sqrt{s}\sqrt{t-s\cos 2\alpha}(t-s)}\exp\left(\frac{|{\bf m}|^2s}{2}\right)
\left(\frac{r_0^2}{2s}\frac{t-s}{t-s\cos 2\alpha}\right)^{\frac{\pi}{2 \alpha}} ({\rm by \ Lemma \ \ref{Bessellemma}})\\
& \leq & \displaystyle \frac{C (t-s)^{\frac{\pi}{2 \alpha}-1}}{[s(t-s \cos 2\alpha)]^{\frac{1}{2}+\frac{\pi}{2 \alpha}}} \exp\left(\frac{|{\bf m}|^2s}{2}\right)\\
& \leq & \displaystyle \frac{C}{{u^*}^{1+\frac{\pi}{ \alpha}} }\exp\left(\frac{|{\bf m}|^2u}{2}\right) (t-s)^{\frac{\pi}{2 \alpha}-1}.
\end{array}
\end{equation}
The first inequality in the RHS of Eq. (\ref{ineq}) holds by first considering
$$
\begin{array}{lll}
&&f(r,s)\pi(r\sin \alpha, t-s)\\
& \leq & b(r,s)\exp\left(m_1(r\cos \alpha-z^{(1)}_0)+m_2(r \sin \alpha-z^{(2)}_0)-\frac{|{\bf m}|^2s}{2}\right) \times\\
&& \displaystyle \frac{r \sin \alpha}{\sqrt{2 \pi (t-s)^3}} \exp \left( -\frac{r^2 \sin^2 \alpha}{2(t-s)}-m_2 r\sin \alpha \right)\\
& = & \displaystyle b(r,s) \frac{r \sin \alpha}{\sqrt{2 \pi (t-s)^3}} \exp\left(-\frac{r^2 \sin^2 \alpha}{2(t-s)}\right)\exp \left( m_1 r\cos \alpha\right) \exp\left(-m_1 z^{(1)}_0-m_2 z^{(2)}_0-\frac{|{\bf m}|^2s}{2}\right)\\
& \leq &\displaystyle  b(r,s)\frac{r \sin \alpha}{\sqrt{2 \pi (t-s)^3}}\exp\left( -\frac{r^2 \sin^2 \alpha}{2(t-s)}\right) \times\\
&& \displaystyle \exp\left(\frac{r^2}{4s}+ s(m_1 \cos \alpha)^2\right)\exp\left(-m_1 z^{(1)}_0-m_2 z^{(2)}_0-\frac{|{\bf m}|^2s}{2}\right)\\
& \leq &\displaystyle  b(r,s)\frac{r \sin \alpha}{\sqrt{2 \pi (t-s)^3}}\exp\left( -\frac{r^2 \sin^2 \alpha}{2(t-s)}+\frac{r^2}{4s}\right) \exp\left(-m_1 z^{(1)}_0-m_2 z^{(2)}_0 + \frac{|{\bf m}|^2s}{2}\right)\\
&=&\displaystyle  \frac{\sqrt{\pi} \sin \alpha}{\sqrt{2(t-s)^3} \alpha^2 s}\exp\left(-\frac{r^2_0}{2s}\right)\exp\left(-m_1 z^{(1)}_0-m_2 z^{(2)}_0+
\frac{|{\bf m}|^2s}{2}\right)\\
&& \times \displaystyle  \exp\left( -\frac{r^2}{4s}\frac{t-s \cos 2 \alpha}{t-s}\right) \sum_{n=1}^{\infty} n \sin \left(\frac{n \pi \tilde{\theta}_0}{\alpha}\right) I_{\frac{n \pi}{\alpha}}(\frac{rr_0}{s})\\
&=& \displaystyle  \frac{C}{\sqrt{(t-s)^3} s}\exp\left(-\frac{r^2_0}{2s}+\frac{|{\bf m}|^2s}{2}\right)
 \displaystyle  \exp\left( -\frac{r^2}{4s}\frac{t-s \cos 2 \alpha}{t-s}\right) \sum_{n=1}^{\infty} n \sin \left(\frac{n \pi \tilde{\theta}_0}{\alpha}\right) I_{\frac{n \pi}{\alpha}}(\frac{rr_0}{s}).
\end{array}
$$
By using the following identity (Abramowitz and Stegun  \cite{AS}),
$$
\int_0^{\infty} e^{-b t^2} I_v(at) dt =\frac{1}{2} \sqrt{\frac{\pi}{b}}\exp\left(\frac{a^2}{8b}\right)I_{\frac{v}{2}}\left(\frac{a^2}{8b}\right),
$$
$$
\begin{array}{lll}
g(s,t) &=&\displaystyle \int_0^{\infty} f(r,s)\pi(r\sin \alpha, t-s) dr\\
 &\leq &\displaystyle \frac{C}{\sqrt{s}\sqrt{t-s\cos 2\alpha}(t-s)}\exp\left(\frac{|{\bf m}|^2s}{2}\right) \exp\left(-\frac{r^2_0 \sin^2 \alpha}{t-s \cos 2\alpha}\right)\times \\
&& \displaystyle \sum_{n=1}^{\infty} n \sin \left(\frac{n \pi \tilde{\theta}_0}{\alpha} \right)
I_{\frac{n \pi}{2\alpha}}\left(\frac{r_0^2}{2s}\frac{t-s}{t-s\cos 2\alpha }\right)\\
 &\leq &\displaystyle \frac{C}{\sqrt{s}\sqrt{t-s\cos 2\alpha}(t-s)}\exp\left(\frac{|{\bf m}|^2s}{2}\right)
\sum_{n=1}^{\infty} n \sin \left(\frac{n \pi \tilde{\theta}_0}{\alpha} \right)
I_{\frac{n \pi}{2\alpha}}\left(\frac{r_0^2}{2s}\frac{t-s}{t-s\cos 2\alpha }\right).
\end{array}
$$
Combining Eq. (\ref{firin}) and Eq. (\ref{ineq}),
there exists $N_0 \in \mathcal{N}^+$, such that for $n > N_0$,
$$
g(s, u+\eta_n) \leq f(u, u^*, s),
$$
where $f(u, u^*, s)=$
$$
\left\{
\begin{array}{ll}
\frac{\sqrt{m^2_2(u+u_1)+2}}{\sqrt{2 \pi} (u- u^*)}\int_0^{\infty} f(r,s) dr,& s \in (0, u^*],\\
\frac{C}{{u^*}^{1+\frac{\pi}{\alpha}} }\exp\left(\frac{|{\bf m}|^2u}{2}\right) \left(1_{\{\alpha \leq \pi/2\}}(u+u_1-u^*)^{\frac{\pi}{2 \alpha}-1}+ 1_{\{\alpha > \pi/2\}}(u-s)^{\frac{\pi}{2 \alpha}-1}\right), & s \in (u^*, u),
\end{array}\right.
$$
is integrable on $(0,u)$ as a function of $s$.}
The DCT implies,
$$
\lim_{n \rightarrow \infty}\int_0^u g(s, u+\eta_n) ds =\int_0^u \lim_{n \rightarrow \infty}g(s, u+\eta_n) ds=\int_0^u g(s, u) ds.
$$
The result follows.
Similarly $\pi$ is continuous in $(0, \infty)\times (0, \infty)$ we have
$$
\frac{1}{u_n}\int_u^{u+u_n}\pi(z_0^{(2)}, v)dv =\pi(z_0^{(2)}, u+\xi_n) \rightarrow \pi(z_0^{(2)}, u) \quad {\rm as} \quad n \rightarrow \infty,
$$
where $0<\xi_n < u_n$.

For the second term of the RHS of Eq. (\ref{nn}),
by Lemma \ref{gcontinuous}, when $n \rightarrow \infty$,
$$
\frac{1}{u_n}\int_u^{u+u_n} g(s, v) dv = g(s, u+\psi_n) \rightarrow g(s, u),
$$
where $0 < \psi_n < u_n$.}
\end{proof}

{Before proving the main result of the subsection, we introduce the following result which is needed in the proof.}
\begin{lemma}[Aven \cite{Aven}]\label{aven}
Let $N(t)$ be a counting process adapted to $\{\mathcal{F}_t\}$ and assume there exists a process $\{\lambda(t)\}$ such that
$
\lambda(t)=\lim_{h_n \downarrow 0} Y_n
$
where $Y_n=E[N(t+h_n)-N(t) \mid \mathcal{F}_t]/h_n$,
if the following conditions hold with $\{\lambda(t)\}$ and $\{y(t)\}$ non-negative and measurable processes:\\
\begin{itemize}
\item for each $t$,  $\lim_{n \rightarrow \infty}Y_n=\lambda(t)$,\\
\item for each $t$ there exists for almost all $\omega$ and $n_0=n_0(t, \omega)$, such that
$|Y_n(s, \omega)-\lambda(s, \omega)|\leq y(s, \omega), s\leq t, n \geq n_0$,\\
\item $\int_{0}^t y(s) ds < \infty, 0 \leq t <\infty$.\\
\end{itemize}
Then $\{N(t)-\int_0^t \lambda(s) ds\}$ is an $\mathcal{F}_t$ martingale.
\end{lemma}

\begin{theorem}\label{thm3}
The intensity process of $\tau_2$ is given by
$$
\lambda_2(u)
=1_{\{\tau_2>u\}} \cdot 1_{\{\tau_1>u\}} \left(\frac{\pi(z_0^{(2)},u)-\int_0^u g(s, u) ds }{\int_{\Omega_{\alpha}} h(r, \theta, u) dr d\theta}\right)
+1_{\{\tau_2 >u\}} \cdot 1_{\{\tau_1 =s\}} \left(\frac{ g(s, u)}{\int_u^{\infty} g(s, v)dv}\right),
$$
where $t_0 < u< t_1$.
\end{theorem}
\begin{proof}
Lemma \ref{localRate} gives a limit of the local default rate, for $0<s<u$,
$$
\begin{array}{ll}
& \displaystyle \lim_{n \rightarrow \infty}
\frac{1_{\{\tau_2 >u\}}}{u_n}P(\tau_2 \in (u, u+u_n] \mid \mathcal{F}_n)\\
=&
\displaystyle 1_{\{\tau_2>u\}}\cdot 1_{\{\tau_1>u\}} \left(\frac{\pi(z_0^{(2)},u)-\int_0^u g(s, u) ds }{\int_{\Omega_{\alpha}} h(r, \theta, u) dr d\theta}\right)
+1_{\{\tau_2 >u\}}\cdot 1_{\{\tau_1 =s\}} \left(\frac{ g(s, u)}{\int_u^{\infty} g(s, v)dv}\right).
\end{array}
$$
To prove this limit is indeed the intensity of $\tau_2$, we need to verify the three conditions in Lemma \ref{aven}.
{
For $\omega \in {\{\tau_2>u, \tau_1> u\}}$ and $0< u <t$, let
$$
Y_n(u, \omega)
=\frac{1}{u_n}\left(\frac{\int_u^{u+u_n}\pi(z_0^{(2)},v)dv-\int_u^{u+u_n}\int_0^u g(s, v) ds dv}{\int_{\Omega_{\alpha}} h(r, \theta, u) dr d\theta}\right).
$$
By Remark \ref{piFunction},
$$
\pi(z_0^{(2)}, u+\xi_n) \leq \pi\left(z_0^{(2)}, \sqrt{\frac{\left(z_0^{(2)}\right)^2}{m_2^2}+\frac{9}{4m_2^4}}-\frac{3}{2m_2^2}\right)
 {}_{=}^{\rm def} \phi(z_0^{(2)}),
$$
and $P(\tau > u)=\int_{\Omega_{\alpha}} h(r, \theta, u) dr d\theta$ is decreasing in $u$, hence,
$$
Y_n(u, \omega) \leq  \frac{\pi(z_0^{(2)},u+\xi_n)}{\int_{\Omega_{\alpha}} h(r, \theta, u) dr d\theta} \leq \frac{\phi(z_0^{(2)})}{\int_{\Omega_{\alpha}} h(r, \theta, t) dr d\theta} {}_{=}^{\rm def} \gamma(u, \omega).
$$
For $\omega \in \{\tau_1=s, \tau_2>u\}$ and $0<u<t$, let
$$
Y_n(u, \omega)= \frac{1_{\{u <s\}}}{u_n}\left(\frac{\int_u^{u+u_n}\pi(z_0^{(2)},v)dv-\int_u^{u+u_n}\int_0^u g(s, v) ds dv}{\int_{\Omega_{\alpha}} h(r, \theta, u) dr d\theta}\right) + \frac{1_{\{u\geq s\}}}{u_n} \left(\frac{\int_u^{u+u_n} g(s, v) dv}{\int_u^{\infty} g(s, v)dv}\right).
$$
One can check that
$$
Y_n(u, \omega) \leq 1_{\{u <s\}}\left(\frac{\pi(z_0^{(2)},u+\xi_n)}{\int_{\Omega_{\alpha}} h(r, \theta, t) dr d\theta}\right)
+ 1_{\{u\geq s\}} \left(\frac{g(s, u+\eta_n) }{\int_t^{\infty} g(s, v)dv}\right).
$$
Let $u^{**}(\omega) \in (s, t)$, for $u \in [u^{**}, t)$,
{similar to Eq. (\ref{firin})},
\begin{equation}\label{gsteq}
g(s, u+\eta_n)  \leq  \frac{\sqrt{m^2_2(u+u_1)+2}}{\sqrt{2 \pi} (u^{**}- s)}\int_0^{\infty} f(r,s) dr.
\end{equation}
We can select $u^{**}(\omega) \in (s, t)$ to be sufficiently close to $s$ and $n_0(\omega) \in \mathcal{N}^+$ such that for $n>n_0(\omega)$ and
$u \in (s, u^{**})$,
$$
\frac{r_0^2}{2s}\left(\frac{u+\eta_n-s}{u+\eta_n-s\cos 2\alpha }\right) \in (0, 1)
$$
and
$$
u+u_n-s\cos 2\alpha =u+u_n-s +2s\sin^2\alpha \geq 2 s \sin^2 \alpha,
$$
which is a positive constant.
{Also similar to Eq. (\ref{ineq}),}
$$
g(s, u+\eta_n) \leq \frac{C}{s^{1+\frac{\pi}{ \alpha}} }\exp\left(\frac{|{\bf m}|^2s}{2}\right) (u+u_n-s)^{\frac{\pi}{2 \alpha}-1}.
$$
Therefore, $g(s, u+\eta_n) \leq \tilde{f}(u, u^{**},s)$
where $\tilde{f}(u, u^{**}, s) =$
{
$$
\left\{
\begin{array}{ll}
\displaystyle \frac{\sqrt{m^2_2(u+u_1)+2}}{\sqrt{2 \pi} (u^{**}- s)}\int_0^{\infty} f(r,s) dr,& u \in [u^{**}, t),\\
\displaystyle \frac{C}{s^{1+\frac{\pi}{ \alpha}} }\exp\left(\frac{|{\bf m}|^2s}{2}\right) \left(1_{\{\alpha \leq \pi/2\}}(u+u_1-s)^{\frac{\pi}{2 \alpha}-1}+ 1_{\{\alpha > \pi/2\}}(u-s)^{\frac{\pi}{2 \alpha}-1}\right), & u \in (s, u^{**}).
\end{array}
\right.
$$}
{Thus, we obtain}
$$
Y_n(u, \omega) \leq 1_{\{u <s\}} \left(\frac{\phi(z_0^{(2)})}{ \int_{\Omega_{\alpha}} h(r, \theta, t) dr d\theta}\right)
+ 1_{\{u\geq s\}} \left(\frac{\tilde{f}(u, u^{**}, s) }{\int_t^{\infty} g(s, v)dv}\right) {}_{=}^{\rm def} \gamma(u, \omega).
$$
Since
$$
\lim_{n \rightarrow \infty} Y_n(u, \omega)=\lambda_2(u, \omega),
$$
we obtain that $\lambda_2(u, \omega) \leq \gamma(u, \omega)$.
One can check that for $0< t \leq t_1$,
$$
\int_0^{t} \gamma(u, \omega) du  < \infty, a.s.
$$
Then all three conditions can be verified.}
\end{proof}

Theorem {\ref{thm3}} gives the intensity of $\tau_2$ when $t_0 < u < t_1$. Theorem {\ref{thm4}} follows immediately when multiply observations are included.

\section{Further Work and Conclusion}\label{Conclusion}

One  possible extension of our model
is to incorporate more than two names.
We assume there are three firms and the asset value of firm $i$ is given by
$V_i(t)$, $i=1,2,3$.
Let $X_i(t)=\ln \frac{V_i(t)}{B_i}$,
where $B_i$ is the default threshold value of
firm $i$.
Then the default time of name $i$ is given by
$$
\tau_i :=\inf\{t>0: X_i(t) \leq 0\}.
$$
Assume that $X_i(t)$, $i=1, 2, 3$, follows the stochastic differential equation
\begin{equation}
d{\bf X}= \mbox{\boldmath$\mu$} dt+ \Sigma d{\bf W},
\end{equation}
where ${\bf W}$ is a two-dimensional standard Brownian motion,
$$
{\bf X}(t)=\left(
\begin{array}{c}
X_1(t)\\
X_2(t)\\
X_3(t)
\end{array}
\right), \quad
\mbox{\boldmath$\mu$}=\left(
\begin{array}{c}
\mu_1\\
\mu_2\\
\mu_3
\end{array}
\right), \quad
\Sigma=
\left(
\begin{array}{cc}
\sigma_1 \sqrt{1-\rho^2}& \sigma_1 \rho\\
0 & \sigma_2\\
\sigma_3 \sqrt{1-\bar{\rho}} & \sigma_3 \bar{\rho}
\end{array}
\right)
$$
and  $X_i(0) > 0$, $i=1,2,3$,
 and $\rho, \bar{\rho}$ ($\bar{\rho}\geq \rho$) characterize the correlation structure of the three names. Denote by
$$
T=
\left(
\begin{array}{ccc}
\sigma_1\sqrt{1-\rho^2} & \sigma_1 \rho & 0\\
0 & \sigma_2 & 0\\
0 &  0& \sigma_3
\end{array}
\right)^{-1}
$$
and define ${\bf Z}=T{\bf X}$, ${\bf m}=T {\mbox{\boldmath$\mu$}}$,
$$
\sigma=T{\Sigma}=\left(
\begin{array}{ccc}
1 &0\\
0 & 1\\
\sqrt{1-\bar{\rho}} &  \bar{\rho}
\end{array}
\right).
$$
Then
$$
d {\bf Z}= {\bf m}dt+ \sigma d{\bf W}.
$$
The equivalent default times can be redefined as
$$
\left\{
\begin{array}{ll}
 \tau_1&=\inf\left\{t>0: Z_1(t)=-\frac{\rho}{\sqrt{1-\rho^2}}Z_2(t)\right\}\\
\tau_2&=\inf\left\{t>0: Z_2(t)=0\right\}\\
\tau_3&=\inf\left\{t>0: Z_3(t)=0\right\}.
\end{array}
\right.
$$
Alternatively,
we have
$$
\left\{
\begin{array}{l}
d Z_1 = m_1 dt+ dW_1\\
d Z_2 = m_2 dt+ dW_2\\
d Z_3 =
(m_3-\sqrt{1-\bar{\rho}^2}m_1-\bar{\rho}m_2) dt
+ \sqrt{1-\bar{\rho}^2}dZ_1 +\bar{\rho}dZ_2.
\end{array}
\right.
$$
If we  assume that
$$
m_3 > \sqrt{1-\bar{\rho}^2}m_1+\bar{\rho}m_2
$$
and
$$
Z_3(0) > \sqrt{1-\bar{\rho}^2}Z_1(0)+\bar{\rho}Z_2(0),
$$
then
the default time $\tau_3> \min\{\tau_1, \tau_2\}$ and the model has practical value. For instance,
one can regard firm 1 as the reference entity
of an insurance contract, firm 2 as the
insurance company who provides protection
to the contract buyer once the reference name defaults,
firm 3 as a reinsurance company that is securer than
firms 1 and 2 and provide a protection to the contract buyer if any of them defaults.

Under these assumptions, one can apply the information setting
in Section 2 and study the intensity processes of the default names by the similar method since we can transform the three name case into
two name case by first consider firms 1 and 2
and then the remaining two names if one of them defaults.
The difficulty of the extension is that
we need to take into account the value of $(Z_1, Z_2)$ at the first default time $\min(\tau_1, \tau_2)$ when considering the remaining two surviving names.
We leave this for further research.

In summary, we present a continuous time structural asset value model
describing the asset value of two firms
driven by correlated Brownian motions and with incomplete information.
We show that the original structural model can be transformed
into a reduced-form intensity-based model.
We derive the conditional distribution of the
default time and that of the asset value of each name.
Furthermore, we derive the explicit form of intensity processes
of the two correlated names  and demonstrate the valuation method of the default intensity in some special cases.
Numerical experiments on the default intensity show that,
the default intensity in the correlated case is nearly
the same as that in independent case, when default is not observed.
Once a default occurred, the default intensity has a sharp change
and this impact decreases gradually.

\vspace{3mm}

\noindent
{\bf Acknownledgement:} This research work was supported by
Research Grants Council of Hong Kong under Grant
Number 17301214 and HKU CERG Grants
and Hung Hing Ying Physical Research Grant.

\end{document}